\documentclass[journal]{IEEEtran}
\usepackage{float}
\usepackage{cite}
\usepackage{amsmath,amssymb,amsfonts,mathrsfs,amsthm}
\usepackage{algorithmic}
\usepackage{graphicx,color}
\usepackage{textcomp}
\usepackage{hyperref}

\newtheorem{assumption}{Assumption}
\newtheorem{lemma}{Lemma}
\newtheorem{remark}{Remark}
\newtheorem{definition}{Definition}
\newtheorem{theorem}{Theorem}

\ifCLASSINFOpdf

\else

\fi

\begin{document}

\title{Privacy-Preserving Distributed Defense Framework for DC Microgrid Against Exponentially Unbounded False Data Injection Attacks}

\author{Yi Zhang,~\IEEEmembership{Student Member,~IEEE,} Mohamadamin Rajabinezhad,~\IEEEmembership{Student Member,~IEEE,} Yichao Wang, 

Junbo Zhao,~\IEEEmembership{Senior Member,~IEEE,}, Shan Zuo,~\IEEEmembership{Member,~IEEE,}
\thanks{Yi Zhang,  Mohamadamin Rajabinezhad, Yichao Wang, Junobo Zhao, and Shan Zuo are with the Department of Electrical and Computer Engineering, University of Connecticut, CT 06269, USA.; (E-mails: yi.2.zhang@uconn.edu;mohamadamin.rajabinezhad@uconn.edu;yichao.wang-
@uconn.edu;junbo@uconn.edu;shan.zuo@uconn.edu)}
}
\markboth{Journal of \LaTeX\ Class Files,~Vol.~14, No.~8, August~2015}%
{Shell \MakeLowercase{\textit{et al.}}: Bare Demo of IEEEtran.cls for IEEE Journals}

\maketitle

\begin{abstract}
This paper introduces a novel, fully distributed control framework for DC microgrid (MG), enhancing resilience against exponentially unbounded false data injection (EU-FDI) attacks. Our framework features a consensus-based secondary control for each converter, effectively addressing these advanced threats. To further safeguard sensitive operational data, a privacy-preserving mechanism is incorporated into the control design, ensuring that critical information remains secure even under adversarial conditions. Rigorous Lyapunov stability analysis confirms the framework's ability to maintain critical DC MG operations like voltage regulation and load sharing under EU-FDI threats. The framework's practicality is validated through hardware-in-the-loop experiments, demonstrating its enhanced resilience and robust privacy protection against the complex challenges posed by quick variant FDI attacks.
\end{abstract}

\begin{IEEEkeywords}
Attack-resilient control, DC microgrid, distributed control, exponentially unbounded attack, FDI attack.
\end{IEEEkeywords}

%
\IEEEpeerreviewmaketitle

\section{INTRODUCTION}
In recent years, DC microgid (MG) have gained prominence over AC MG due to their compatibility with the direct current characteristics of emerging distributed energy resources (DERs), storage units, and controllable loads \cite{1}. DC MG streamline control by eliminating issues like frequency regulation, reactive power, and harmonics found in AC systems \cite{am5,am20}. Key control objectives include voltage regulation to maintain constant average voltage and proportional load sharing based on converter power ratings \cite{am38,9}. A hierarchical control scheme—comprising primary, secondary, and tertiary levels—efficiently addresses these objectives. The primary level employs droop-based methods for rapid voltage and current regulation, while the secondary level compensates voltage deviations caused by primary control \cite{am21,am15}. Distributed secondary control, leveraging local controllers and information, offers a scalable, reliable alternative to centralized approaches, minimizing communication complexity \cite{7,yang2022pi, sadabadi2022robust}.

DC MG face heightened vulnerability to cyberattacks due to their reliance on distributed control frameworks. Cooperative control, involving information exchange among converters, local sensing, and decentralized droop mechanisms, increases exposure to networked control threats over sparse communication networks \cite{liu2011false,kosut2011malicious,liang2016review}. Moreover, safeguarding operational data privacy is critical in power markets, where voltage and load information can be exploited to predict energy pricing or manipulate the market. Unauthorized access to such data may enable adversaries to gain unfair financial advantages and disrupt market stability. These risks highlight the need for privacy-preserving mechanisms to secure sensitive information while maintaining system performance.

To counter cyber threats, various attack-detection methods and resilient control protocols have been developed for DC MG \cite{26,28,37}. Recent advancements focus on distributed control strategies to address unbounded cyberattacks \cite{lu2022generalized,am29,zuo2020resilient,leng2022projections,zhou2023distributed,zhou2022resilient,jiang2021high,liu2023resilient,wang2022discrete}. Input-output feedback linearization enhances secondary voltage control by linking output dynamics to control inputs, but uncontrolled input changes can destabilize the system, particularly under aggressive time-varying attacks that exploit saturation constraints \cite{leng2022projections}. Methods like the generalized extended state observer-based control in \cite{lu2022generalized} estimate unbounded FDI attacks, assuming the attack signal can be detected and its derivative is bounded.

Recent approaches to address privacy concerns in DC MG include optimized energy sharing, state decomposition with noise masking, aggregated data strategies, limited data exchange, and decomposed scheduling methods \cite{huang2016shepherd,yuan2023distributed,hussain2016resilient,zhou2020privacy,albaker2018privacy}. \cite{huang2016shepherd} preserves privacy in hybrid AC-DC MG through the Shepherd framework, which hides power consumption via optimized energy sharing and reduced neighbor transmission. \cite{yuan2023distributed} ensures privacy in islanded DC MG using state decomposition to separate control signals and random noise masking to protect transmitted data. \cite{hussain2016resilient} enhances privacy in networked MG with a nested energy management strategy that aggregates surplus and deficit power data to minimize customer data exposure. \cite{zhou2020privacy} achieves privacy in islanded reconfigurable MG by using a distributed control strategy that shares only frequency data while keeping generation data private. \cite{albaker2018privacy} protects privacy in integrated MG through a Lagrangian relaxation method, decomposing scheduling to safeguard MG data and reduce operation costs.

This paper addresses the challenges of exponentially unbounded false data injection (EU-FDI) attacks in DC MG by introducing a novel distributed secondary control framework. Using adaptive control, the framework ensures uniformly ultimately bounded (UUB) stability, maintaining frequency and voltage regulation under fast-growing attack scenarios. Unlike existing methods, it tackles the complex dynamics of EU-FDI attacks while integrating privacy-preserving mechanisms to protect sensitive data, such as voltages and error variables. This dual-focus approach combines robust defense with stringent privacy, offering a comprehensive solution for modern DC MG. The key contributions are summarized as follows.
\begin{itemize}
    \item[1)] A time-varying coupling gain is introduced and dynamically adjusted through an adaptive control law. Additionally, an adaptive damping mechanism is proposed to balance response speed and stability. These mechanisms enable the system to counteract the rapid variations of EU-FDI attacks while achieving robust stability performance. 
    \item[2)] 
Privacy-preserving mechanisms are integrated into the distributed attack-resilient secondary control framework. A novel output mask is introduced to protect agents' initial states, avoiding the use of random noise and ensuring that all agents converge exactly to the average value of their initial states, rather than its mean square value. This approach achieves privacy preservation while ensuring attack resilience.
    \item[3)]
A rigorous Lyapunov stability analysis is presented, providing theoretical guarantees for the framework's resilience and demonstrating uniformly ultimately bounded (UUB) convergence of the system. The proposed approach is validated through hardware-in-the-loop (HIL) experiments, highlighting its effectiveness in ensuring stable operation under varying load conditions, communication failures, and EU-FDI attacks.
\end{itemize}

The paper is organized as follows: Section \ref{sec:Standard Cooperative Secondary Control} proposes the standard cooperative secondary control for DC MG. Section \ref{sec:Attack-Resilient Privacy-Preserving Control} formulates attack-resilient control problems for bounded and unbounded attacks and presents a fully distributed solution. Section \ref{sec:Hardware-in-the-Loop Validation} validates the results through hardware-in-the-loop (HIL) experiments, and Section \ref{sec:Conclusion} concludes the paper.

\noindent\textit{Notations:}
In this paper, ${{\mathbf{1}}_N} \in {\mathbb{R}^N}$ represents a vector in which every entry is one. The notation $\left|\cdot \right|$ denotes the absolute value of a real number. Additionally, $\operatorname{diag} \left\{  \cdot  \right\}$ is used to form a diagonal matrix from a given set of elements. A physical, islanded DC MG system is modeled as a communication digraph $\mathscr{G}$. This DC MG system comprises $N$ converters. The interactions among local converters are represented by the graph ${\mathscr{G}} = (\mathscr{W}, \mathscr{E}, \mathcal{A})$, where $\mathscr{W}=\{0,1,2, \ldots, N\}$ is the set of vertices, $0$ denotes the leader, and $1\cdots N $ denote the followers. $\mathscr{E} \subset \mathscr{W} \times \mathscr{W}$ is the set of edges, and $\mathcal{A} = [{a_{ij}}]$ is the adjacency matrix. An edge in the graph, representing the information flow from converter $j$ to converter $i$, is denoted by $({w_j}, {w_i})$ and is assigned a weight of ${a_{ij}}$. If $({w_j}, {w_i}) \in \mathscr{E}$, then ${a_{ij}} > 0$; otherwise, ${a_{ij}} = 0$. A node $j$ is considered a neighbor of node $i$ if $({w_j}, {w_i}) \in \mathscr{E}$. The set of neighbors for node $i$ is defined as ${\mathcal{N}_i} = \left\{j \mid ({w_j}, {w_i}) \in \mathscr{E}\right\}$. The system includes one leader node, while the other nodes are followers, which may be subject to FDI attacks. The leader node disseminates the reference value to those converters that can receive its information. Each converter directly receives relative information from linked converters in a sparse communication digraph. The in-degree matrix $\mathcal{D} = \operatorname{diag}({d_i}) \in {\mathbb{R}^{N \times N}}$, is defined with ${d_i} = \sum\nolimits_{j \in {\mathcal{N}_i}} {{a_{ij}}}$, representing the sum of the weights of the edges incident to node $i$. The Laplacian matrix of the graph is given by $\mathcal{L} = \mathcal{D} - \mathcal{A}$. We assume that ${\mathscr{G}}$ is bidirectional, implying ${a_{ij}} = {a_{ji}}$, which makes $\mathcal{L}$ not only semi-positive but also symmetric. The pinning gain $g_i$ represents the influence from the leader to converter $i$: $g_i > 0$ if there is a link from the leader node to node $i$, and $g_i = 0$ otherwise. The pinning gain matrix is defined as ${\mathcal{G}} = \operatorname{diag}(g_i)$.

\section{Standard Cooperative Secondary Control}
\label{sec:Standard Cooperative Secondary Control}

In the standard secondary control, each converter transmits $\mathrm{X}_i = [\bar{V}_i, R_i^{\text{vir}} I_i]$ to its neighbors via a communication graph, where $\bar{V}_i$ is the estimated average voltage, $I_i$ is the output current, and $R_i^{\text{vir}}$ is a virtual impedance set as $R_i^{\text{vir}} = k / I_i^{\text{rated}}$. Since $R_i^{\text{vir}} I_i = k I_i / I_i^{\text{rated}}$, achieving proportional load sharing reduces to achieving consensus on $R_i^{\text{vir}} I_i$.

Based on \cite{zuo2020distributed}, the standard cooperative secondary control for DC MG transforms the problem into consensus control for first-order linear MAS. Its primary objectives are to regulate the average voltage to a global reference and ensure proportional load sharing. Using relative information from neighboring converters, the secondary control adjusts the local voltage setpoint $V_i^{}$. At the primary level, the droop mechanism employs a virtual impedance $R_i^{\text{vir}}$ to model the converter’s output impedance. Cooperative secondary control fine-tunes $V_i^{}$ and mitigates voltage and current residuals. The local voltage setpoint is given as
\begin{equation}
\label{eq1}
V_i^{*}=V_{n_i}+V_{\mathrm{ref}}-R_i^{\mathrm{vir}} I_i
\end{equation}
where $V_{n_i}$ is the reference for the primary control level and is selected at the secondary control level.

To achieve global voltage regulation and proportional load sharing, the secondary control locally provides $V_{n_i}$ for each converter through data exchange with neighbors. By applying input-output feedback linearization, the voltage droop mechanism in \eqref{eq1} is differentiated to obtain
\begin{equation}
\label{eq2}
\dot{V}_{n_i}=\dot{V}_i^{*}+R_i^{\mathrm{vir}} \dot{i}_i=u_i,
\end{equation}
where $u_i$ is the control input, and \eqref{eq2} computes the primary control reference $V_{n_i}$ from $u_i$. The secondary control is reformulated as a consensus problem for first-order linear MAS. To achieve global voltage regulation and proportional load sharing, the cooperative secondary control law for each converter, based on relative information from neighbors, is given by
\begin{equation}
\label{eq3}
u_i=c_i\left(g_i\left(V_{\mathrm{ref}}-\bar{V}_i\right)+\sum_{j \in \mathcal{N}_i} a_{i j}\left(R_{j}^{\mathrm{vir}} I_{j}-R_i^{\mathrm{vir}} I_i\right)\right),
\end{equation}
where $c_i \in \mathbb{R}>0$ is the coupling gain. $\bar{V}_i$ is the estimate of the global average voltage value at converter $i$ and is given by
\begin{equation}
\label{eq4}
\dot{\bar{V}}_i=\dot{V}_i+c_i \sum_{j \in \mathcal{N}_i} a_{i j}\left(\bar{V}_{j}-\bar{V}_i\right)
\end{equation}
where $V_i$ is the local measured voltage. The secondary control setpoint for the primary control, $V_{n_i}$, is then computed from $u_i$ as
\begin{equation}
\label{eq5}
V_{n_i}=\int u_i \mathrm{~d} t
\end{equation}

Assume that the converter produces the demanded voltage, i.e., $V_i^{*}=V_i$. Combining \eqref{eq2}, \eqref{eq3}, and \eqref{eq4} yields
\begin{equation}
\label{eq6}
\begin{aligned}
\dot{\bar{V}}_i+R_i^{\mathrm{vir}} \dot{i}_i=c_i & \left(\sum_{j \in \mathcal{N}_i} a_{i j}\left(\bar{V}_{j}-\bar{V}_i\right)+g_i\left(V_{\text {ref }}-\bar{V}_i\right)\right. \\
& \left.+\sum_{j \in \mathcal{N}_i} a_{i j}\left(R_{j}^{\mathrm{vir}} I_{j}-R_i^{\mathrm{vir}} I_i\right)\right),
\end{aligned}
\end{equation}
which are further formulated as
\begin{flalign}
\label{eq7}
\begin{aligned}
\dot{\bar{V}}_i+R_i^{\mathrm{vir}} \dot{i}_i=&c_i\Bigg( \sum_{j \in \mathcal{N}_i} a_{i j}\left(\left(\bar{V}_{j}+R_{j}^{\mathrm{vir}} I_{j}\right)-\left(\bar{V}_i+R_i^{\mathrm{vir}} I_i\right)\right) \\
& +g_i\left(\left(V_{\mathrm{ref}}+R_i^{\mathrm{vir}} I_i\right)-\left(\bar{V}_i+R_i^{\mathrm{vir}} I_i\right)\right)\Bigg)
\end{aligned}&&\raisetag{2\baselineskip}
\end{flalign}

Reference \cite{nasirian2014distributed} gives the detailed steady-state analysis to show that the cooperative secondary control achieves both voltage and current regulation objectives, due to the relationship between the supplied currents and the bus voltage through the MG admittance matrix. Similarly, we obtain that, in the steady state, $R_i^{\text {vir }} I_i$ converges to a certain constant value $k I_{s s}^{\mathrm{pu}}$. Denote $\Theta_i=\bar{V}_i+R_i^{\mathrm{vir}} I_i$ and $\Theta_{\mathrm{ref}}=V_{\mathrm{ref}}+k I_{s s}^{\mathrm{pu}}$. Then,
\begin{equation}
\label{eq8}
\begin{aligned}
\dot{\Theta}_i & =c_i\left(\sum_{j \in \mathcal{N}_i} a_{i j}\left(\Theta_{j}-\Theta_i\right)+g_i\left(\Theta_{\mathrm{ref}}-\Theta_i\right)\right) \\
& =c_i\left(-\left(d_i+g_i\right) \Theta_i+\sum_{j \in \mathcal{N}_i} a_{i j} \Theta_{j}+g_i \Theta_{\mathrm{ref}}\right) .
\end{aligned}
\end{equation}

The global form of \eqref{eq8} is
\begin{equation}
\label{eq9}
\dot{\Theta}=-\operatorname{diag}\left(c_i\right)(\mathcal{L}+\mathcal{G})\left(\Theta-\mathbf{1}_{N} \Theta_{\mathrm{ref}}\right),
\end{equation}
where $\Theta=\left[\Theta_{1}^{\mathrm{T}}, \ldots, \Theta_{N}^{\mathrm{T}}\right]^{\mathrm{T}}$.

Define the following global cooperative regulation error
\begin{equation}
\label{eq10}
\varepsilon=\Theta-\mathbf{1}_{N} \Theta_{\mathrm{ref}}
\end{equation}
where $\varepsilon=\left[\varepsilon_{1}^{\mathrm{T}}, \ldots, \varepsilon_{N}^{\mathrm{T}}\right]^{\mathrm{T}}$. The following assumption is needed for the communication network.
\begin{assumption}
\label{ass:1}
The digraph $\mathscr{G}$ includes a spanning tree, where the leader node is the root. \end{assumption}
\begin{lemma}[\cite{fax2004information}]
\label{lem:1}
Given Assumption 1, $(\mathcal{L}+\mathcal{G})$ is nonsingular and positive-definite.    
\end{lemma}

\begin{lemma}[\cite{zuo2020distributed}]
\label{lem:1}
Given Assumption 1, by designing the auxiliary control input as \eqref{eq3} and \eqref{eq4}, the global voltage regulation and proportional load sharing are both achieved.    
\end{lemma}
\begin{remark}
\label{rem:1}
Using input-output feedback linearization, the secondary control problem of DC MG is transformed into a tracking synchronization problem for first-order linear MAS. By applying standard cooperative control protocols, we proposed the secondary control protocols \eqref{eq3} and \eqref{eq4} to achieve global voltage regulation and proportional load sharing, offering a faster dynamic response than the double PI controllers in \cite{nasirian2014distributed}.
\end{remark}

\section{Attack-Resilient Privacy-Preserving Control}
\label{sec:Attack-Resilient Privacy-Preserving Control}

In this section, the attack-resilient privacy-preserving control problem for DC MG under EU-FDI attacks is formulated. Then, an attack-resilient privacy-preserving control framework to
address the problem is developed. Rigorous proofs based on Lyapunov
techniques show that UUB convergence is achieved
for voltage and current regulations against EU-FDI attacks.

\subsection{Problem Formulation}

Malicious attackers may inject EU-FDI attack signals to the local control input channel of each converter. Hence, instead of \eqref{eq2}, for DC MG under EU-FDI attacks, the DC MG have
\begin{equation}
\label{eq11}
{{\dot V}_{n_i}} = \dot V_i^* + R_i^{\operatorname{vir} }{{\dot I}_i} = {\bar u}_i = {u_i} + {\delta _i},
\end{equation}
where ${{\bar u}_i}$ represents the corrupted input signal, while ${\delta _i}$ denotes the potential FDI attack injections targeting the local control input channel. The attackers' objective is to destabilize the cooperative regulation system by injecting these FDI attack.

\begin{definition}[Polynomially Unbounded Attack]
\label{def:1}
The attack signals $\delta_i\in C^{\gamma}$ are polynomially unbounded with the finite polynomial order of $\gamma$. That is, $|\delta_i^{(\gamma)}| \leq \kappa_i$, where $\gamma > 0$ is a scalar, and $\kappa_i$ is a positive constant.
\end{definition}
\begin{remark}
\label{rem:2}

Definition \ref{def:1} addresses a wider range of unbounded FDI attacks compared to \cite{liu2023resilient,zuo2022adaptive,zhou2023distributed}, relaxing the requirement from bounded first-time derivatives to bounded higher-order derivatives. Depending on the defender's computational power, the highest order $\gamma$ can be significantly large. Polynomially unbounded attacks on the control input may cause rapid variations in controlled variables, potentially leading to system instability due to saturation.
\end{remark}

\begin{definition}[EU-FDI Attack]
\label{def:2}
Let $\delta_i$ represent a cyberattack at time \(t\), where \(t \geq 0\) and \(i\) indexes a particular attribute or effect of the attack (such as the number of systems compromised, data volume stolen, etc.). An exponential cyberattack satisfies $\left| \delta_i \right| \leq e^{\kappa_i t}$
for some positive constant \(\kappa_i\), where \(e\) is the base of natural logarithms. 
\end{definition}

\begin{remark}
\label{rem:3}
This definition asserts that the impact of the attack grows exponentially with respect to time, with \(\kappa_i\) determining the rate of exponential growth. The condition \(\left| \delta_i \right| \leq e^{\kappa_i t}\) ensures that the impact does not exceed the exponential function \(e^{\kappa_i t}\) at any time \(t\).
The attackers' injections ${\delta _i}$ can be any EU-FDI attacks signals. Compable with the resilient control protocols for DC MG in \cite{lu2022generalized} and \cite{am29}, which deal with bounded noises/disturbances, our research addresses the emerging and realistic threat of unbounded attack injections in quantum-influenced cyber environments. The projection of cyberattack signals increasing at an exponential rate reflects a plausible risk in the quantum era, where limitations on the power of attack signals are increasingly less defined.  
\end{remark}

Use \eqref{eq9} and \eqref{eq10}, one has 
\begin{equation}
\label{eq12}
\dot{\varepsilon} = -\operatorname{diag}(c_i)(\mathcal{L} + \mathcal{G})\varepsilon    
\end{equation}

Consider the attack injections in \eqref{eq11} and use the standard secondary control protocols \eqref{eq3} and\eqref{eq4}, based on \eqref{eq12}, the error dynamics under EU-FDI attacks is
\begin{equation}
\label{eq12}
\dot{\varepsilon}=-\operatorname{diag}\left(c_i\right)(\mathcal{L}+\mathcal{G}) \varepsilon+\delta 
\end{equation}
where $\delta=\left[\delta_{1}^{\mathrm{T}}, \ldots, \delta_{N}^{\mathrm{T}}\right]^{\mathrm{T}}$ is the attack vector. Since $\delta$ is exponentially unbounded, $\varepsilon$ is also exponentially unbounded. That is, the standard secondary control fails to preserve the stability of the DC MG in the presence of EU-FDI attacks. It is hence important to develop advanced attack-resilient control approach to address such unbounded attacks for MG.

To evaluate the convergence results of the attack-resilient method to be designed, the following definition is introduced.
\begin{definition}[\cite{khalil2002nonlinear}]
\label{def:3}
$x(t) \in \mathbb{R}$ is UUB with the ultimate bound $b$, if there exist constants $b, c>0$, independent of $t_{0} \geq 0$, and for every $a \in(0, c)$, there exists $t_{1}=t_{1}(a, b) \geq 0$, independent of $t_{0}$, such that
\begin{equation}
\label{eq13}
\left|x\left(t_{0}\right)\right| \leq a \Rightarrow|x(t)| \leq b, \quad \forall t \geq t_{0}+t_{1}
\end{equation}    
\end{definition}

Besides, privacy preservation ensures that the true initial states of the DC MG remain concealed from adversaries during consensus computation. To protect the initial states of each DC MG, an output mask function is introduced to obscure the internal states $\bar{V}_i$ and $I_i$, defined as follows.
\begin{definition}[\cite{altafini2019dynamical}]
\label{def:4}
The function \( h_i(t, \bar{V}_i, \pi_i) \) is said an output mask with the privacy-preserving property for agent \( i \) if it is local and in addition:
\begin{itemize}
    \item[C1:] \hypertarget{cond:C1}{\( h_i(0, \bar{V}_i, \pi_i) \neq \bar{V}_i ,\forall \bar{V}_i \in \mathbb{R}^n, i = 1, 2, \dots, N; \)}
    \item[C2:] \hypertarget{cond:C2}{\( h_i(t, \bar{V}_i, \pi_i) \) guarantees indiscernibility of the initial conditions;}
    \item[C3:] \hypertarget{cond:C3}{\( h_i(t, \bar{V}_i, \pi_i) \) does not preserve neighborhoods of any \( \forall \bar{V}_i \in \mathbb{R}^n; \)}
    \item[C4:] \hypertarget{cond:C4}{\( h_i(t, \bar{V}_i, \pi_i) \) strictly increases in \( \bar{V}_i \) for each fixed \( t \) and \( \pi_i, i = 1, 2, \dots, N; \)}
    \item[C5:] \hypertarget{cond:C5}{\( |h_i(t, \bar{V}_i, \pi_i) - \bar{V}_i| \) is decreasing in \( t \) for each fixed \( \bar{V}_i \) and \( \pi_i \), and \( \lim_{t \to \infty} h_i(t, \bar{V}_i, \pi_i) = \bar{V}_i, i = 1, 2, \dots, n. \)}
\end{itemize}
\end{definition}

\begin{remark}
\label{rem:4}
Definition \ref{def:4} ensures privacy protection for the initial states of DC MG. Condition \hyperref[cond:C1]{C1} guarantees that the output mask's initial state differs from the agent's actual state, preventing disclosure. Conditions \hyperref[cond:C2]{C2} and \hyperref[cond:C3]{C3} further support privacy requirements. Condition \hyperref[cond:C4]{C4}, a generalized form of the $K_{\infty}$ function, ensures that $h_i\left(t, \bar{V}_i, \pi_i\right)$ is a bijection in $x$ for fixed $t$ and $\pi$ without preserving the origin. To enhance privacy, the mask can asymptotically converge to the true state, satisfying Condition \hyperref[cond:C5]{C5}. When $h_i\left(t, \bar{V}_i, \pi_i\right)$ meets Definition \ref{def:4}, the initial state $\bar{V}_i(0)$ remains protected.
\end{remark}

Now, we formulate the attack-resilient privacy-preserving control problems for DC MG against EU-FDI attacks while achieve privacy-preserving.

\noindent\textbf{Problem} (Attack-Resilient Privacy-Preserving Control Problem)\textbf{.}
\textit{Under the EU-FDI attacks on local control input channels, design local control protocols $u_i$ in \eqref{eq11} for each converter using only the local measurement such that, for all initial conditions, $\varepsilon$ in \eqref{eq12} is UUB while achieve privacy-preserving. That is, the bounded global voltage regulation and proportional load sharing are both achieved.}

\subsection{Attack-Resilient Privacy-Preserving Controller Design}

Based on the Definition \ref{def:4}, the following output masks for each agent is constructed
\begin{align}
\label{eq14}\phi_i(t)=&f_i\left(t, \bar{V}_i, \pi_i\right)
\\\label{eq15}\psi_i(t)=&h_i\left(t, I_i, \pi_i\right)
\\\label{eq16}f_i\left(t, \bar{V}_i, \pi_i\right) =&\left(1+\lambda_i e^{-\rho_i t}\right)\left(\bar{V}_i(t)+\gamma_i e^{-\theta_i t}\right)
\\\label{eq17}h_i\left(t, I_i, \pi_i\right) =&\left(1+\lambda_i e^{-\mu_i t}\right)\left(I_i(t)+\gamma_i e^{-\epsilon_i t}\right)
\end{align}
where $\phi_i(t)$ and $\psi_i(t)$ are the output states of the output mask $h_i\left(t, \bar{V}_i, \pi_i\right)$ and $h_i\left(t, I_i, \pi_i\right)$ for agent $i$, respectively; and $\lambda_i$, $\gamma_i$, $\rho_i$, $\theta_i$, $\mu_i$, and $\epsilon_i$ are scalars.
\begin{lemma}[\cite{altafini2019dynamical}]
\label{lem:3}
The function $h_i\left(t, \bar{V}_i, \pi_i\right)=$ $\left(1+\lambda_i e^{-\rho_i t}\right)\left(\bar{V}_i+\gamma_i e^{-\theta_i t}\right)$ is a privacy output mask to mask the internal state $\bar{V}_i$.
\end{lemma}
\begin{lemma}[\cite{altafini2019dynamical}]
\label{lem:4}
The function $h_i\left(t, I_i, \pi_i\right)=$ $\left(1+\lambda_i e^{-\mu_i t}\right)\left(I_i+\gamma_i e^{-\epsilon_i t}\right)$ is a privacy output mask to mask the internal state $I_i$.
\end{lemma}

According to Lemma \ref{lem:3} and \ref{lem:4}, the privacy preservation for initial states of agents in DC MG can be guaranteed.


Next, the detailed formulation and components of our distributed exponentially attack-resilient privacy-preserving controller is introduced. To begin, we define
\begin{equation}
\label{eq18}
\begin{split}
\zeta_i =c_i\left( {\sum\limits_{j \in {\mathcal{N}_i}} {{a_{ij}}\left( \phi_j - \phi_i \right)}  + {g_i}\left( {{V_{\operatorname{ref} }} - \phi_i} \right)} \right.
\\\left. + \sum\limits_{j \in {\mathcal{N}_i}} {{a_{ij}}\left( {R_j^{\operatorname{vir} }{\psi_j} - R_i^{\operatorname{vir} }{\psi_i}} \right)}  \right)
\end{split}
\end{equation}

To ensure bounded global voltage regulation and proportional load sharing under EU-FDI attacks, a comprehensive distributed exponentially attack-resilient privacy-preserving control protocols is proposed, well-designed to enhance the robustness of DC MG against such attacks as follows
\begin{align}
&\label{eq19}\begin{aligned}u_i=c_i\mathrm{exp}\left({\xi_i}\right)\left( {\sum\limits_{j \in {\mathcal{N}_i}} {{a_{ij}}\left( {{\phi_j} - {\phi_i}} \right)}  + {g_i}\left( {{V_{\operatorname{ref} }} - {\phi_i}} \right)} \right.
\\\left. + \sum\limits_{j \in {\mathcal{N}_i}} {{a_{ij}}\left( {R_j^{\operatorname{vir} }{\psi_j} - R_i^{\operatorname{vir} }{\psi_i}} \right)}  \right)
\end{aligned}
\\\label{eq20}&\dot{\xi}_i=\alpha_i\left|\zeta_i\right|-\beta_i\left({\xi}_i-{\hat{\xi}}_i\right)
\\\label{eq21}& \dot{\hat{\xi}}_i=\rho_i\left({\xi}_i-{\hat{\xi}}_i\right)
\end{align}
where $\alpha_i$, $\beta_i$, $c_i$, and $\rho_i$ are positive constants.
\begin{remark}
\label{rem:5}
The equations \eqref{eq19}-\eqref{eq23} define a resilient control framework for mitigating exponentially unbounded FDI (EU-FDI) attacks. Equation \eqref{eq19} ensures dynamic voltage regulation, while \eqref{eq20} employs exponential control inputs and $\zeta_i$-based adjustments to enhance attack resilience. Equations \eqref{eq21} and \eqref{eq22} provide adaptive damping to balance response speed and stability. Together, these mechanisms ensure robust secondary control and stability for DC MG under EU-FDI attacks.
\end{remark}

\subsection{Main Results}

Next, we give the main result of solving the attack-resilient
secondary control problems for DC MG.
\begin{theorem}
\label{thm:1}
Given Assumptions \ref{ass:1}, for DC MG under the EU-FDI attacks in Definition \ref{def:2}, by using the attack-resilient privacy preserving control protocols consist of \eqref{eq18}, \eqref{eq19}, \eqref{eq20} and \eqref{eq21}, the cooperative regulation error $\varepsilon$ in \eqref{eq10} is UUB. That is, the attack-resilient privacy-preserving control problem is solved. 
\end{theorem}

\begin{proof}
Denote $\Phi_i=\phi_i+R_i^{\mathrm{vir}} \psi_i$, for \eqref{eq18}, one has
\begin{equation}
\label{eq22}
\begin{aligned}
\zeta_i =& c_i\left( {\sum\limits_{j \in {\mathcal{N}_i}} {{a_{ij}}\left( \phi_j - \phi_i \right)}  + {g_i}\left( {{V_{\operatorname{ref} }} - \phi_i} \right)} \right.
\\&\left. + \sum\limits_{j \in {\mathcal{N}_i}} {{a_{ij}}\left( {R_j^{\operatorname{vir} }{\psi_j} - R_i^{\operatorname{vir} }{\psi_i}} \right)}  \right)
\\=&c_i\left(\sum_{j \in \mathcal{N}_i} a_{i j}\left(\Phi_{j}-\Phi_i\right)+g_i\left(\Theta_{\mathrm{ref}}-\Phi_i\right)\right) 
\\=&-c_i\left(d_i+g_i\right) \Phi_i+c_i\sum_{j \in \mathcal{N}_i} a_{i j} \Phi_{j}+c_ig_i \Theta_{\mathrm{ref}} .
\end{aligned}
\end{equation}

Based on \eqref{eq22}, the time derivative of $\zeta_i$ is
\begin{equation}
\label{eq23}
\begin{aligned}
\dot{\zeta}_i = &-c_i\left(d_i+g_i\right)\dot{\Phi}_i+c_i\sum_{j \in \mathcal{N}_i} a_{i j} \dot{\Phi}_j
\\= &-c_i\left(d_i+g_i\right) \left(\mathrm{exp}\left({\xi_i}\right)\zeta_i+\delta_i\right)
\\&+c_i\sum_{j \in \mathcal{N}_i} a_{i j} \left(\mathrm{exp}\left({\xi_j}\right)\zeta_j+\delta_j\right) .
\end{aligned}
\end{equation} 

Then the global form of  $\dot{\zeta}_i$ is
\begin{equation}
\label{eq24}
\begin{aligned}
\dot\zeta
=&-\left(\mathcal{L} + \mathcal{G}\right)\left(\operatorname{diag} \left(\mathrm{exp}\left({\xi_i}\right) \right)\zeta +\delta\right)
\end{aligned}
\end{equation}

Consider the following Lyapunov function candidate
\begin{equation}
\label{eq25}
E \left( t \right)= \frac{1}{2}\sum\limits_{i  =1}^{N} \zeta_i^2\mathrm{exp}\left({\xi_i}\right) 
\end{equation}

Based on \eqref{eq23}, its time derivative along the system trajectory of $\zeta$ is given by
\begin{flalign}
\label{eq26}
\begin{aligned}
\dot{E}=& \sum_{i=1}^{N} \left(\zeta_i  \dot{\zeta}_i\mathrm{exp}\left(\xi_i\right)+\frac{1}{2}\zeta_i^2\mathrm{exp}\left(\xi_i\right)\dot{\xi}_i\right) 
\\=&\zeta^{\mathrm{T}} \operatorname{diag}\left(\mathrm{exp}\left(\xi_i\right)\right) \dot{\zeta}
+ \frac{1}{2}\zeta^{\mathrm{T}} \operatorname{diag}\left(\mathrm{exp}\left(\xi_i\right)\right) \operatorname{diag}\left(\dot{\xi}_i\right)\zeta
\\=& -\zeta^{\mathrm{T}} \operatorname{diag}\left(\mathrm{exp}\left(\xi_i\right)\right)
(\mathcal{L}+\mathcal{G})\left(\operatorname{diag}\left(\mathrm{exp}\left(\xi_i\right)\right) \zeta+\delta\right)
\\&+\frac{1}{2}\zeta^{\mathrm{T}} \operatorname{diag}\left(\mathrm{exp}\left(\xi_i\right)\right) \operatorname{diag}\left(\dot{\xi}_i\right)\zeta
\end{aligned}&&\raisetag{2.8\baselineskip}   
\end{flalign}

Based on the norm bound property, for \eqref{eq26}, by using Young's inequality $2 a b \leq q a^{2}+\frac{1}{q} b^{2}$ where $q$ is a positive number, and Sylvester's inequality $0\leq\sigma_{\min}(P)\left\|x\right\|^2\leq x^{\mathrm{T}}Px\leq\sigma_{\max}(P)\left\|x\right\|^2$ where $x$ is a non-zero vector and $P$ is a positive-definite matrix, one has
\begin{flalign}
\label{eq27}
\begin{aligned}
\dot{E}\leq&-\zeta^{\mathrm{T}} \operatorname{diag}\left(\mathrm{exp}\left(\xi_i\right)\right)(\mathcal{L}+\mathcal{G})\operatorname{diag}\left(\mathrm{exp}\left(\xi_i\right)\right) \zeta
\\&-\zeta^{\mathrm{T}} \operatorname{diag}\left(\mathrm{exp}\left(\xi_i\right)\right)(\mathcal{L}+\mathcal{G})\delta
\\&+\frac{1}{2}\left\|\zeta^{\mathrm{T}}\operatorname{diag}\left(\mathrm{exp}\left(\xi_i\right)\right)\right\| \left\|\operatorname{diag}\left(\dot{\xi}_i\right)\zeta\right\|
\\\leq&-\zeta^{\mathrm{T}} \operatorname{diag}\left(\mathrm{exp}\left(\xi_i\right)\right)(\mathcal{L}+\mathcal{G})\operatorname{diag}\left(\mathrm{exp}\left(\xi_i\right)\right) \zeta
\\&-\zeta^{\mathrm{T}} \operatorname{diag}\left(\mathrm{exp}\left(\xi_i\right)\right)(\mathcal{L}+\mathcal{G})\delta
\\&+\frac{1}{2}\left\|\zeta^{\mathrm{T}}\operatorname{diag}\left(\mathrm{exp}\left(\xi_i\right)\right)\right\| \left\|\operatorname{diag}\left(\dot{\xi}_i\right)\zeta\right\|
\\\leq&-\sigma_{\min }\left(\mathcal{L}+\mathcal{G}\right)\left\|\operatorname{diag}\left(\mathrm{exp}\left(\xi_i\right)\right) \zeta\right\|^{2}
\\&+\sigma_{\max }(\mathcal{L}+\mathcal{G})\left\|\operatorname{diag}\left(\mathrm{exp}\left(\xi_i\right)\right) \zeta\right\|\|\delta\|
\\&+\frac{1}{4}\left( \left\|\operatorname{diag}\left(\mathrm{exp}\left(\xi_i\right)\right) \zeta\right\|^{2}+\left\|\operatorname{diag}\left(\dot{\xi}_i\right)\zeta\right\|^2\right)
\\\leq&-\left(\sigma_{\min }(\mathcal{L}+\mathcal{G})-\frac{1}{4}\right)\left\|\operatorname{diag}\left(\mathrm{exp}\left(\xi_i\right)\right) \zeta\right\|
\\&\times\left(\left\|\operatorname{diag}\left(\mathrm{exp}\left(\xi_i\right)\right) \zeta\right\|-\frac{\sigma_{\max }(\mathcal{L}+\mathcal{G})}{\sigma_{\min }(\mathcal{L}+\mathcal{G})-\frac{1}{4}}\|\delta\|\right.
\\&\left.-\frac{1}{4} \frac{\left\|\operatorname{diag}\left(\dot{\xi}_i\right)\zeta\right\|^2}{\left(\sigma_{\min }(\mathcal{L}+\mathcal{G})-\frac{1}{4}\right)\left\|\operatorname{diag}\left(\mathrm{exp}\left(\xi_i\right)\right) \zeta\right\|}\right)
\end{aligned}&&\raisetag{10\baselineskip}
\end{flalign}

Based on \eqref{eq20} and the integral solution for differential equations, one has
\begin{flalign}
\label{eq28}
\begin{split}
\mathrm{exp}\left(\xi_i\right) =\mathrm{exp}\Bigg(\xi_i(0)+\int_0^t\bigg(\alpha_i\left|\zeta_i(\tau)\right|
-\beta_i\bigg({\xi}_i(\tau)\\-{\hat{\xi}}_i(\tau)\bigg)\bigg)\mathrm{d} \tau\Bigg)
\end{split}&&\raisetag{2.3\baselineskip}
\end{flalign}

Define
$\tilde{\xi}_i=\xi_i-\hat{\xi}_i$, based on \eqref{eq20} and \eqref{eq21}, its time derivative is
\begin{equation}
\label{eq29}
\begin{aligned}
\dot{\tilde{\xi}}_i=&\alpha_i\left|\zeta_i\right|-\left(\beta_i+\rho_i\right)\left({\xi}_i-{\hat{\xi}}_i\right)
\\=&\alpha_i\left|\zeta_i\right|-\left(\beta_i+\rho_i\right){\tilde{\xi}}_i.
\end{aligned} 
\end{equation}

For \eqref{eq29}, based on the integral solution for differential equations, one has
\begin{flalign}
\label{eq30}
\begin{split}
{\tilde{\xi}}_i=\mathrm{exp}\left(-\left(\beta_i+\rho_i\right)t\right){\tilde{\xi}}_i\left(0\right)+\alpha_i\int_0^t\left(\mathrm{exp}\bigg(-\left(\beta_i+\rho_i\right)\right.
\\\left.\times\left(t-\tau\right)\bigg)
\left|\zeta_i\left(\tau\right)\right|\right)\mathrm{d}\tau   
\end{split}&&\raisetag{2\baselineskip}
\end{flalign}

Since $\lim\nolimits_{t\to\infty}\mathrm{exp}\left(-\left(\beta_i+\rho_i\right)\left(t-\tau\right)\right)\left|\zeta_i\left(\tau\right)\right|\mathrm{d}\tau=0$, the integral $\int_0^t\mathrm{exp}\left(-\left(\beta_i+\rho_i\right)\left(t-\tau\right)\right)$ $\left|\zeta_i\left(\tau\right)\right|\mathrm{d}\tau$ is UUB. Besides, $\lim\nolimits_{t\to\infty}\mathrm{exp}\left(-\left(\beta_i+\rho_i\right)t\right){\tilde{\xi}}_i\left(0\right)=0$, then the UUB of $\tilde{\xi}_i$ is proved. Assume $\lim_{t\to\infty}\left|\tilde{\xi}_i\right|=\eta$, based on \eqref{eq28}, one has
\begin{flalign}
\label{eq31}
\begin{aligned}
&\lim_{t\to\infty}\mathrm{exp}\left(\xi_i\right)
\\\geq&\lim_{t\to\infty}\mathrm{exp}\left(\xi_i(0)+\int_0^t\left(\alpha_i\left|\zeta_i(\tau)\right|-\beta_i\left|\tilde{\xi}_i(\tau)\right|\right)\mathrm{d} \tau\right)
\\=& \lim_{t\to\infty}\mathrm{exp}\bigg(\xi_i(0)+\left(\alpha_i\left|\zeta_i(t)\right| -\beta_i\left|\tilde{\xi}_i(t)\right|\right)t\bigg)
\\=& \lim_{t\to\infty}\mathrm{exp}\bigg(\xi_i(0)+\left(\alpha_i\left|\zeta_i(t)\right| -\beta_i\eta \right)t\bigg)
\end{aligned}&&\raisetag{3.5\baselineskip}
\end{flalign}
Choosing $\left| \zeta_i \right|\geq \frac{\beta_i\eta}{\alpha_i}$, one has
\begin{flalign}
\label{eq32}
\begin{aligned}
&\lim_{t\to\infty}\frac{\dot{\xi}_i^2}{\left(4\sigma_{\min }(\mathcal{L}+\mathcal{G})-1\right)\mathrm{exp}\left(\xi_i\right)}
\\\leq& \lim_{t\to\infty}\frac{\left(\alpha_i\left|\zeta_i\right|-\beta_i\tilde{\xi}_i\right)^2}{\left(4\sigma_{\min }(\mathcal{L}+\mathcal{G})-1\right)\mathrm{exp}\bigg(\xi_i(0)+\left(\alpha_i\left|\zeta_i(t)\right| -\beta_i\eta \right)t\bigg)}
\\\leq& \lim_{t\to\infty}\frac{\left(\alpha_i\left|\zeta_i\right|-\beta_i\eta\right)^2}{\left(4\sigma_{\min }(\mathcal{L}+\mathcal{G})-1\right)\mathrm{exp}\bigg(\xi_i(0)+\left(\alpha_i\left|\zeta_i(t)\right| -\beta_i\eta \right)t\bigg)}
\\=&0
\end{aligned}&&\raisetag{3.8\baselineskip}
\end{flalign}

Choosing $\xi_i(0)>0$ and $\left| \zeta_i \right|\geq \frac{\beta_i\eta+\kappa_i}{\alpha_i}>\frac{\beta_i\eta}{\alpha_i}$, one has
\begin{equation}
\label{eq33}
\begin{aligned}
&\lim_{t\to\infty}\left(\mathrm{exp}\left(\xi_i\right)- \frac{\dot{\xi}_i^2}{\left(4\sigma_{\min }(\mathcal{L}+\mathcal{G})-1\right)\mathrm{exp}\left(\xi_i\right)}\right)
\\=&\lim_{t\to\infty}\mathrm{exp}\left(\xi_i\right)
\\\geq&\lim_{t\to\infty}\mathrm{exp}\bigg(\xi_i(0)+\left(\alpha_i\left|\zeta_i(t)\right| -\beta_i\eta \right)t\bigg)
\\\geq&\lim_{t\to\infty}\mathrm{exp}\bigg(\xi_i(0)+\kappa_it\bigg)
\\\geq&\lim_{t\to\infty}\mathrm{exp}\left(\kappa_it\right)
\\\geq&\lim_{t\to\infty}\left|\delta_i\right|
\end{aligned}
\end{equation}

Based on \eqref{eq32} and \eqref{eq33}, by choosing $\left|\zeta_i\right|\geq\max \left\{\frac{\beta_i\eta+\kappa_i}{\alpha_i},\frac{\sigma_{\max }(\mathcal{L}+\mathcal{G})}{\sigma_{\min }(\mathcal{L}+\mathcal{G})-\frac{1}{4}} \right\}$, one has 
\begin{equation}
\label{eq34}
\begin{split}
\lim_{t\to\infty}\left(\mathrm{exp}\left(\xi_i\right)- \frac{\dot{\xi}_i^2}{\left(4\sigma_{\min }(\mathcal{L}+\mathcal{G})-1\right)\mathrm{exp}\left(\xi_i\right)}\right)\left|\zeta_i \right|
\\\geq \frac{\sigma_{\max }(\mathcal{L}+\mathcal{G})}{\sigma_{\min }(\mathcal{L}+\mathcal{G})-\frac{1}{4}}\lim_{t\to\infty}\left|\delta_i\right|
\end{split}
\end{equation}

Based on \eqref{eq34}, one has 
\begin{equation}
\label{eq35}
\begin{split}
\lim_{t\to\infty}\left(\mathrm{exp}\left(\xi_i\right)\left|\zeta_i \right|-\frac{\sigma_{\max }(\mathcal{L}+\mathcal{G})}{\sigma_{\min }(\mathcal{L}+\mathcal{G})-\frac{1}{4}}\left|\delta_i\right|\right.
\\\left.- \frac{\dot{\xi}_i^2\left|\zeta_i\right|}{\left(4\sigma_{\min }(\mathcal{L}+\mathcal{G})-1\right)\mathrm{exp}\left(\xi_i\right)}\right)\geq 0
\end{split}
\end{equation}

Since singular values are non-negative real numbers, based on \eqref{eq35}, one has
\begin{equation}
\label{eq36}
\begin{split}
\lim_{t\to\infty}\left(\left\|\operatorname{diag}\left(\mathrm{exp}\left(\xi_i\right)\right) \zeta\right\|-\frac{\sigma_{\max }(\mathcal{L}+\mathcal{G})}{\sigma_{\min }(\mathcal{L}+\mathcal{G})-\frac{1}{4}}\|\delta\|\right.
\\\left.- \frac{\left\|\operatorname{diag}\left(\dot{\xi}_i\right)\zeta\right\|^2}{\left(4\sigma_{\min }(\mathcal{L}+\mathcal{G})-1\right)\left\|\operatorname{diag}\left(\mathrm{exp}\left(\xi_i\right)\right) \zeta\right\|}\right)\geq 0
\end{split}
\end{equation}
Then, based on \eqref{eq27} and \eqref{eq36}, one has $\lim_{t\to\infty}\dot{E}\leq0$ when $\left|\zeta_i\right|\geq\max \left\{\frac{\beta_i\eta+\kappa_i}{\alpha_i},\frac{\sigma_{\max }(\mathcal{L}+\mathcal{G})}{\sigma_{\min }(\mathcal{L}+\mathcal{G})-\frac{1}{4}} \right\}$. Therefore, based on the LaSalle's invariance principle \cite{krstic1995nonlinear}, choosing $\xi_i(0)>0$ and $\left|\zeta_i\right|\geq\max \left\{\frac{\beta_i\eta+\kappa_i}{\alpha_i},\frac{\sigma_{\max }(\mathcal{L}+\mathcal{G})}{\sigma_{\min }(\mathcal{L}+\mathcal{G})-\frac{1}{4}} \right\}$, $\dot E\left( t \right) \leq 0$, for all $
t \geq T $. Hence, the UUB of $\left|\zeta_i\right|$ is proved, and the ultimate bound of $\left|\zeta_i\right|$ is $\max \left\{\frac{\beta_i\eta+\kappa_i}{\alpha_i},\frac{\sigma_{\max }(\mathcal{L}+\mathcal{G})}{\sigma_{\min }(\mathcal{L}+\mathcal{G})-\frac{1}{4}} \right\}$. Since $(\mathcal{L} + \mathcal{G})$ is non-singular and $\left|\zeta_i\right|$ is UUB, based on $\lim_{t\to\infty}\zeta = -\operatorname{diag}(c_i)(\mathcal{L} + \mathcal{G})\varepsilon$, $\varepsilon$ is UUB. This completes the proof.
\end{proof}

\section{Hardware-in-the-Loop Validation}
\label{sec:Hardware-in-the-Loop Validation}

A low-voltage DC MG, with a structure shown in Fig.~\ref{Fig3}, is modeled to study the effectiveness of the proposed control methodology. The practical validation of our control protocol and the DC MG model consists of four DC-DC
converters emulated on a Typhoon HIL 604 system which showed in Fig.~\ref{Fig4}, ensuring a high-fidelity replication of real-world scenarios.
Each source is driven by a buck converter. The converters have similar typologies but different ratings, i.e., the rated currents are equal to $I_{1,2,3,4}^{\operatorname{rated}}=(6,3,3,6)$, besides virtual impedance are equal to $R_{1,2,3,4}^{\operatorname{vir}}=(2,4,4,2)$. 
The converter parameters are $C=2.2 \mathrm{mF}, L=2.64 \mathrm{mH}, f_{s}=60 \mathrm{kHz}, R_{\text {line }}=0.1 \Omega$, $R_{L}=10 \Omega, V_{\text {ref }}=48 \mathrm{~V}$, and $V_{\mathrm{in}}=80 \mathrm{~V}$. 
The rated voltage of the DC MG is $48 V$. The communication network is shown in Fig.~\ref{Fig3}. Communication links are assumed bidirectional to feature a Laplacian matrix and help with the sparsity of the resulting communication graph. 
\begin{figure}[ht]
\centering
\includegraphics[width=0.4\textwidth]{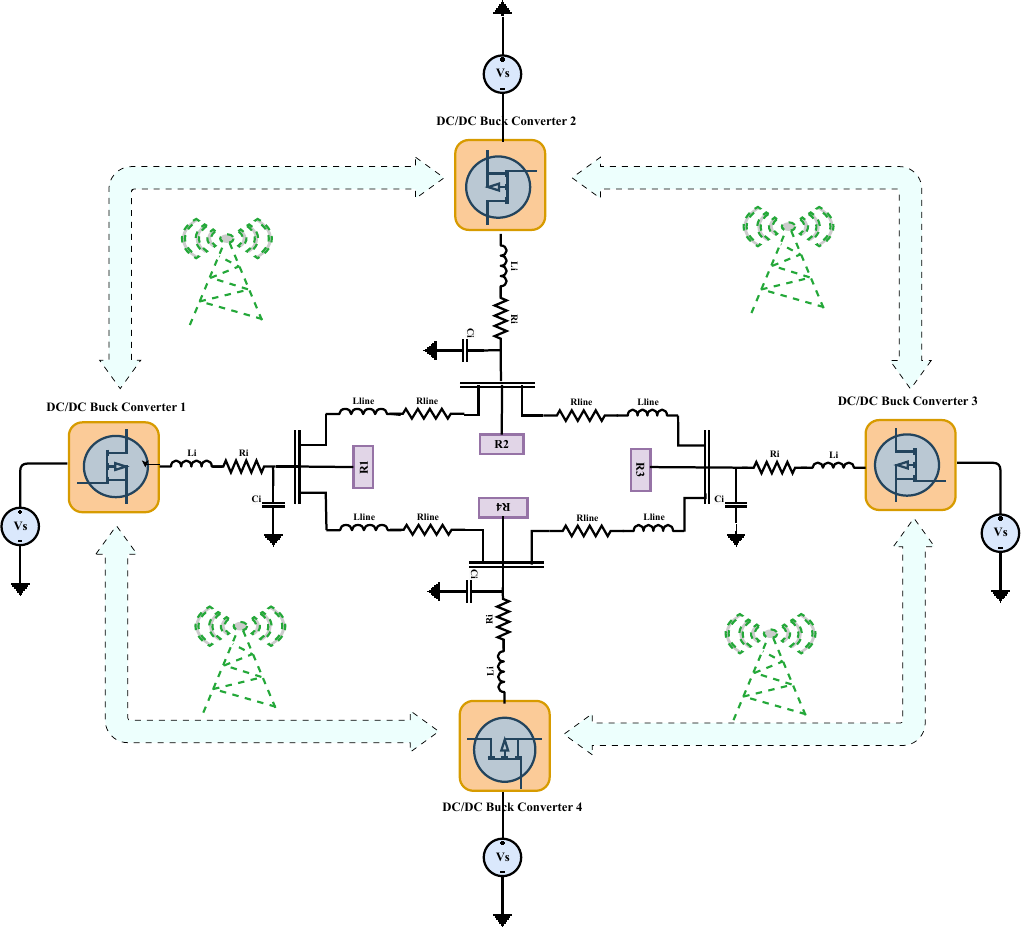}
\caption {The tested DC MG physical structure.}
\label{Fig3}
\end{figure}

In this section, several cases are designed to verify the effectiveness of the proposed controller.

\begin{figure}[ht]
\centering
\includegraphics[width=0.4\textwidth]{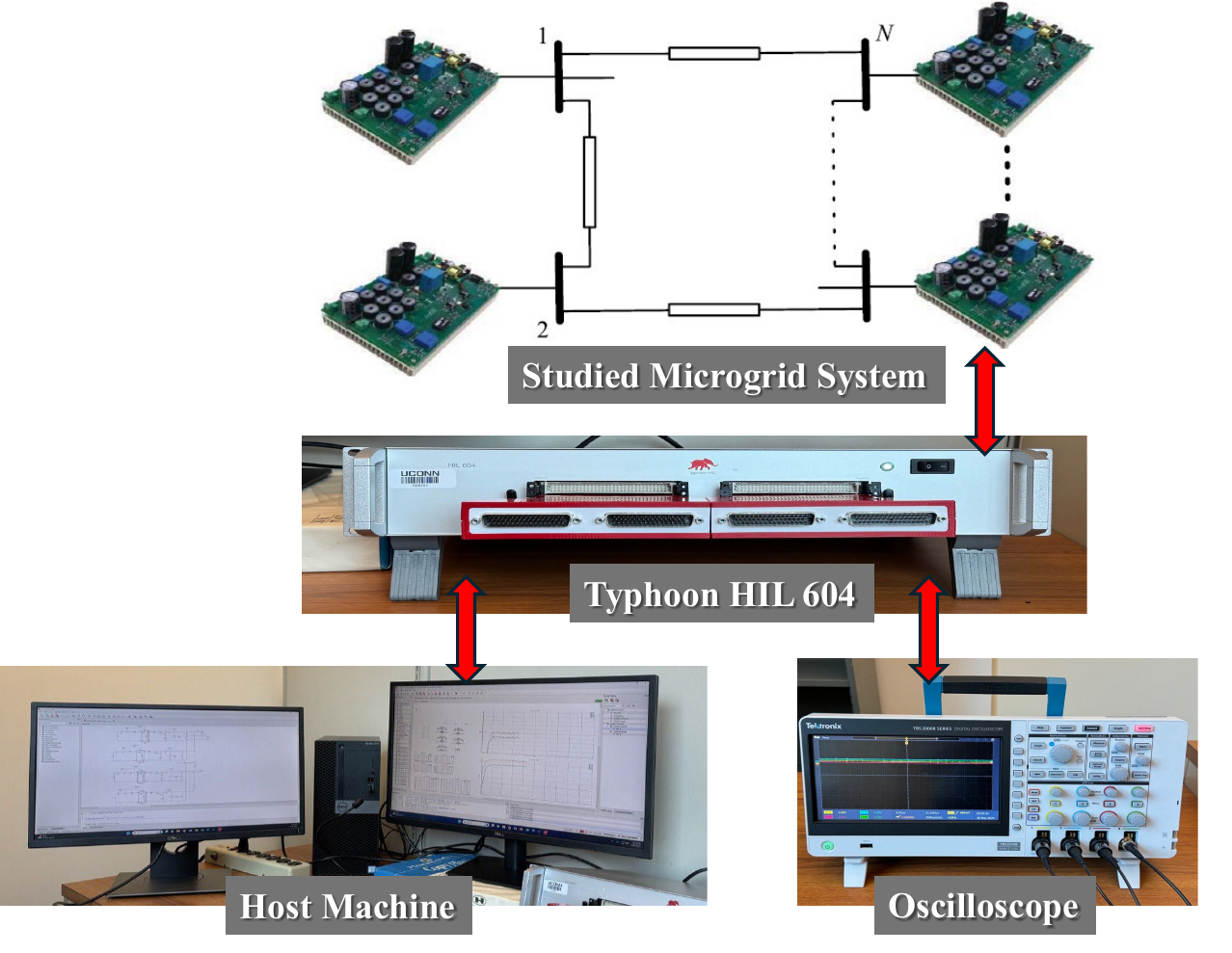}
\caption {The tested DC MG physical structure built using Typhoon HIL devices.}
\label{Fig4}
\end{figure}

\subsection{Voltage Regulation and Current Sharing Test}
In this scenario, the aim is to confirm the effectiveness of the suggested controller in attaining the control objective even when facing a EU-FDI attacks. To assess its performance, a comparison is made with the conventional secondary controller. 
The test lasts from $0$ to $20$ seconds. This part discusses the EU-FDI attack model, which involves injecting EU-FDI attacks at the local control input of each converter by selecting $\delta_i=\begin{bmatrix}3\mathrm{exp}\left(0.1t\right)&4\mathrm{exp}\left(0.2t\right)&0.5\mathrm{exp}\left(0.2t\right)&0.1\mathrm{exp}\left(0.3t\right)\end{bmatrix}^\mathrm{T}$, $\forall i=1,2,3,4 $.  
The adaptive tuning parameters for the resilient control method are set as  $\alpha_i = 1$, $\xi _i(0)=2$, $\hat\xi _i(0)=1.5$, $\forall i = 1, 2, 3, 4$. 
Initially, the conventional secondary controller is illustrated to be ineffective when subjected to an EU-FDI attacks on the MG system. Evidently as shown in Fig.~\ref{FIG5}, following the onset of the FDI attack at approximately $ t=4.8 s$, both bus voltage and current exhibit a continuous rise, indicating the incapacity of the traditional controller to fulfill control objectives in the presence of such attacks.

\begin{figure}[ht]
\centering
{\includegraphics[width=0.4\textwidth]{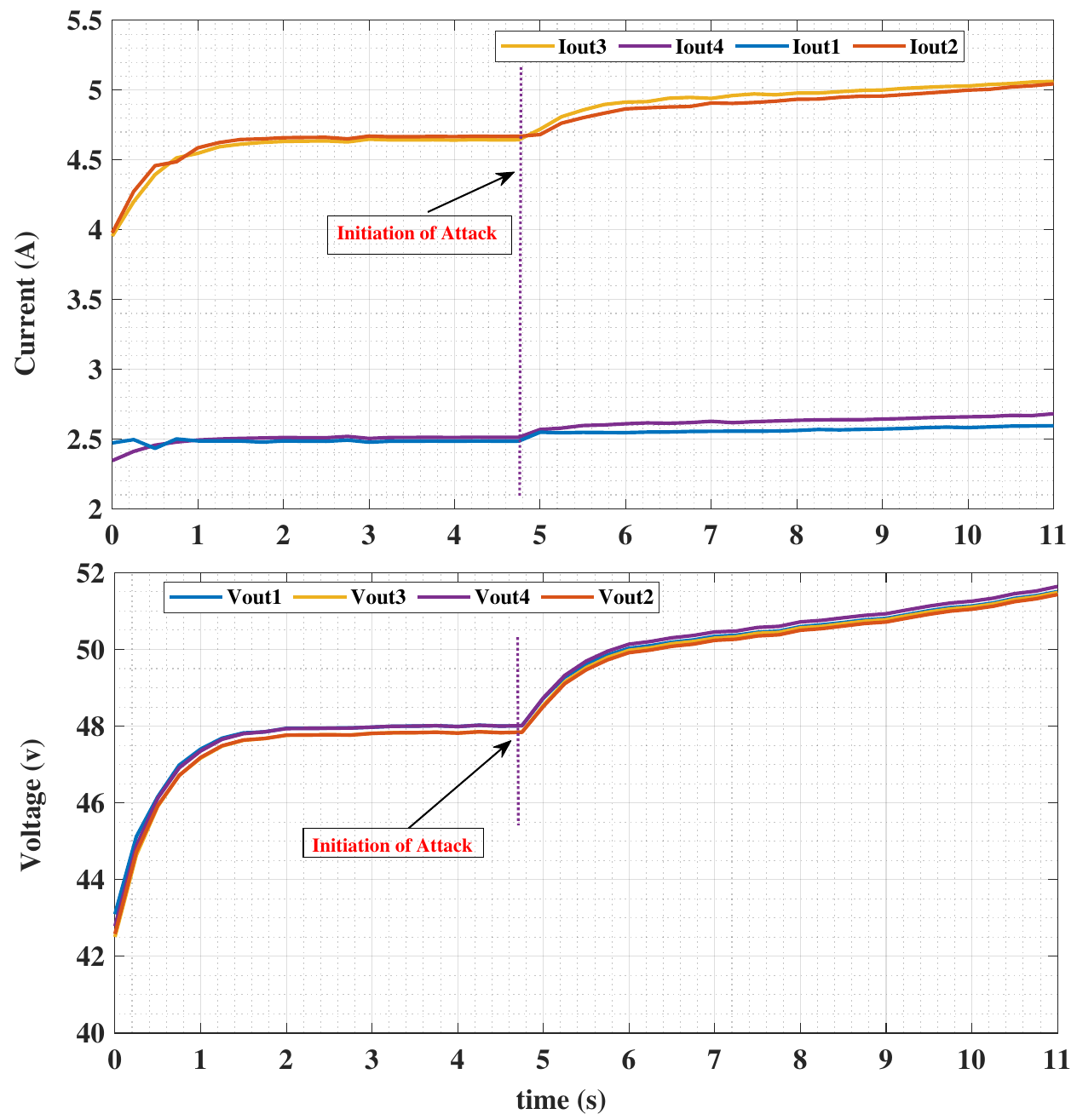}}
\caption {Performance of the conventional control approach in the case of unbounded attack: Supplied currents (top), Terminal voltages of Converters (bottom).}
\label{FIG5}
\end{figure}

\begin{figure}[ht]
\centering
{\includegraphics[width=0.4\textwidth]{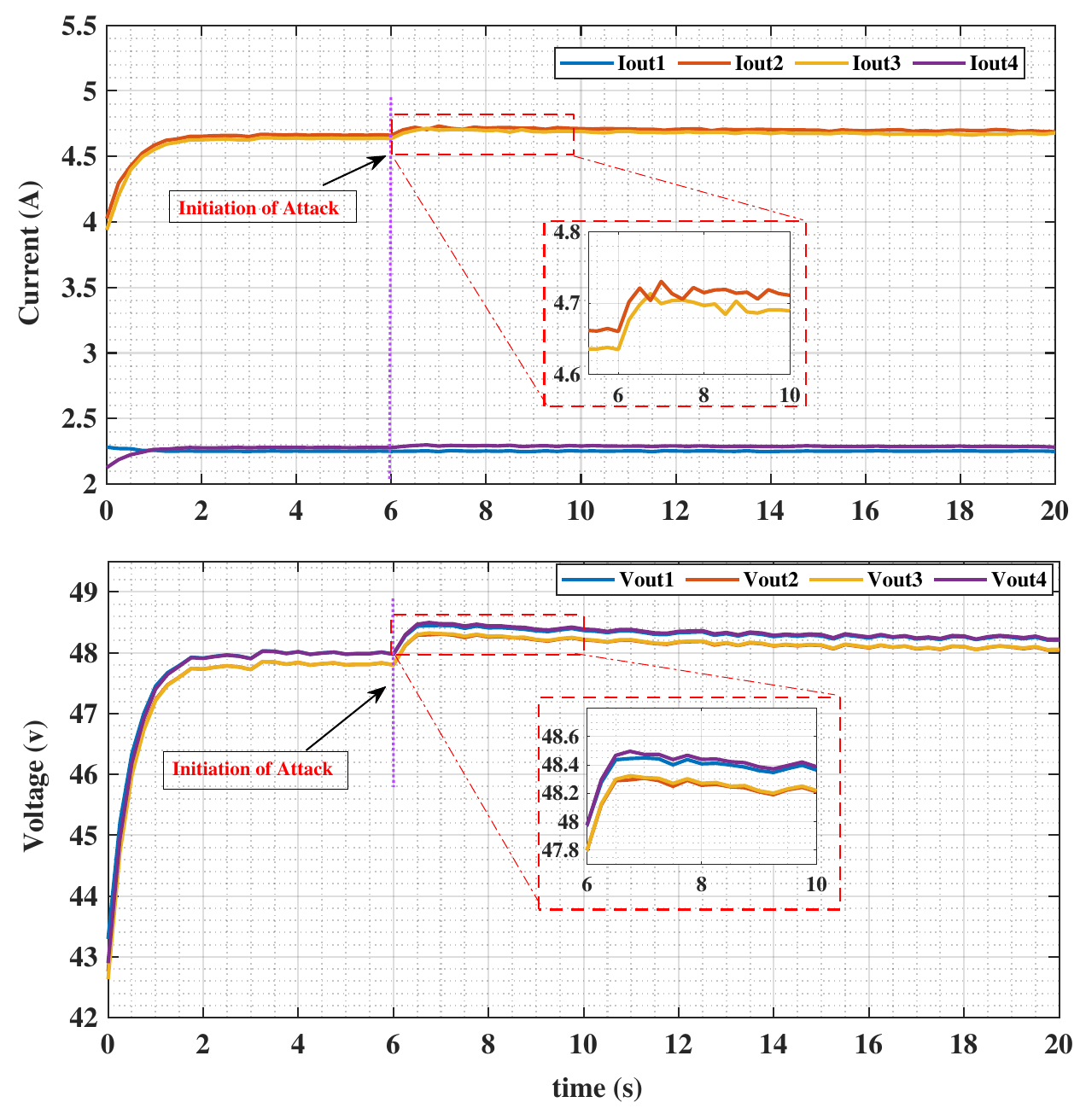}}
\caption {Performance of the proposed attack-resilient control approach in the case of unbounded attack $\delta_i=\left[3~\mathrm{exp}\left(0.1t\right),4~\mathrm{exp}\left(0.2t\right),0.5~\mathrm{exp}\left(0.2t\right),0.1~\mathrm{exp}\left(0.3t\right)\right]^\mathrm{T},\forall i=1,2,3,4 $: Supplied currents (top); Terminal voltages of Converters (bottom).}
\label{FIG6}
\end{figure}

Conversely, Fig.~\ref{FIG6} demonstrates that by using the proposed resilient control method, the terminal voltages of the converters stay bounded and remain within a small neighborhood of the desired value of $48 V$. Additionally this figure shows that the supplied currents are properly shared despite the EU-FDI attacks.
When an attacker injects an unbounded attack signal into the system, it forces the power sources to supply additional energy to the MG. This extraneous energy perturbs the system's equilibrium, leading to deviations in voltage and current from their original states.

Our attack-resilient protocol is designed to maintain system stability and ensure that the voltage and current remain within safe operational limits. However, due to the additional energy introduced by the attack, the system can only achieve a new steady state where the voltage and current are bounded but not necessarily identical to their pre-attack values.

In essence, the unbounded attack injects energy into the system, which alters the energy balance. While our protocol effectively mitigates the impact of the attack and prevents catastrophic failures, it cannot remove the excess energy introduced. This results in a new equilibrium state rather than a complete return to the original states of voltage and current.

To summarize, the reason the voltage and current cannot recover to their original states after an attack is due to the extra energy injected into the system by the attack signal. Our protocol ensures that the system remains stable and the variables are bounded, but the energy imbalance caused by the attack means that a return to the exact pre-attack conditions is not possible.

This case studies verify the effectiveness of the proposed resilient approach in solving mentioned issues, i.e. proportional load sharing and voltage regulation under EU-FDI attacks.

\subsection{Resilient Controller in 
\texorpdfstring{\cite{liu2023resilient}}{Liu 2023}
}
In this scenario, the performance of the secondary controller proposed in \cite{liu2023resilient} is tested under EU-FDI attacks. The experimental outcomes are depicted in Fig.~\ref{FIG70}, where all DGs experience unbounded attacks mentioned in the previous case study. Fig.~\ref{FIG70} illustrates that the terminal voltages of the converters fail to return to the reference value, and the current cannot be regulated to a predefined proportion. Consequently, the entire system experiences a collapse, rendering it unstable and incapable of continued operation. This observation leads to the conclusion that in the presence of an attacker injecting unbounded false data, the distributed algorithm proposed in \cite{liu2023resilient} exhibits complete failure.
In this study, \cite{liu2023resilient} was selected for comparison as it serves as a representative method capable of addressing bounded false data injection (FDI) attacks, where the first derivative of the attack signal is bounded. However, \cite{liu2023resilient} is not designed to handle Exponentially Unbounded Attacks (EU-FDI), making it an appropriate baseline to demonstrate the advancements of our proposed framework. This comparison highlights the limitations of existing approaches, such as \cite{liu2023resilient}, in dealing with more sophisticated and challenging attack scenarios. By addressing the specific challenges posed by exponentially growing attack signals, our framework represents a significant step forward in enhancing system stability and resilience. The superiority of our method is evidenced by the experimental results in Section V, which clearly illustrate its effectiveness in mitigating the impact of EU-FDI attacks, a problem that has not been addressed in prior studies.

\begin{figure}[ht]
\centering
{\includegraphics[width=0.4\textwidth]{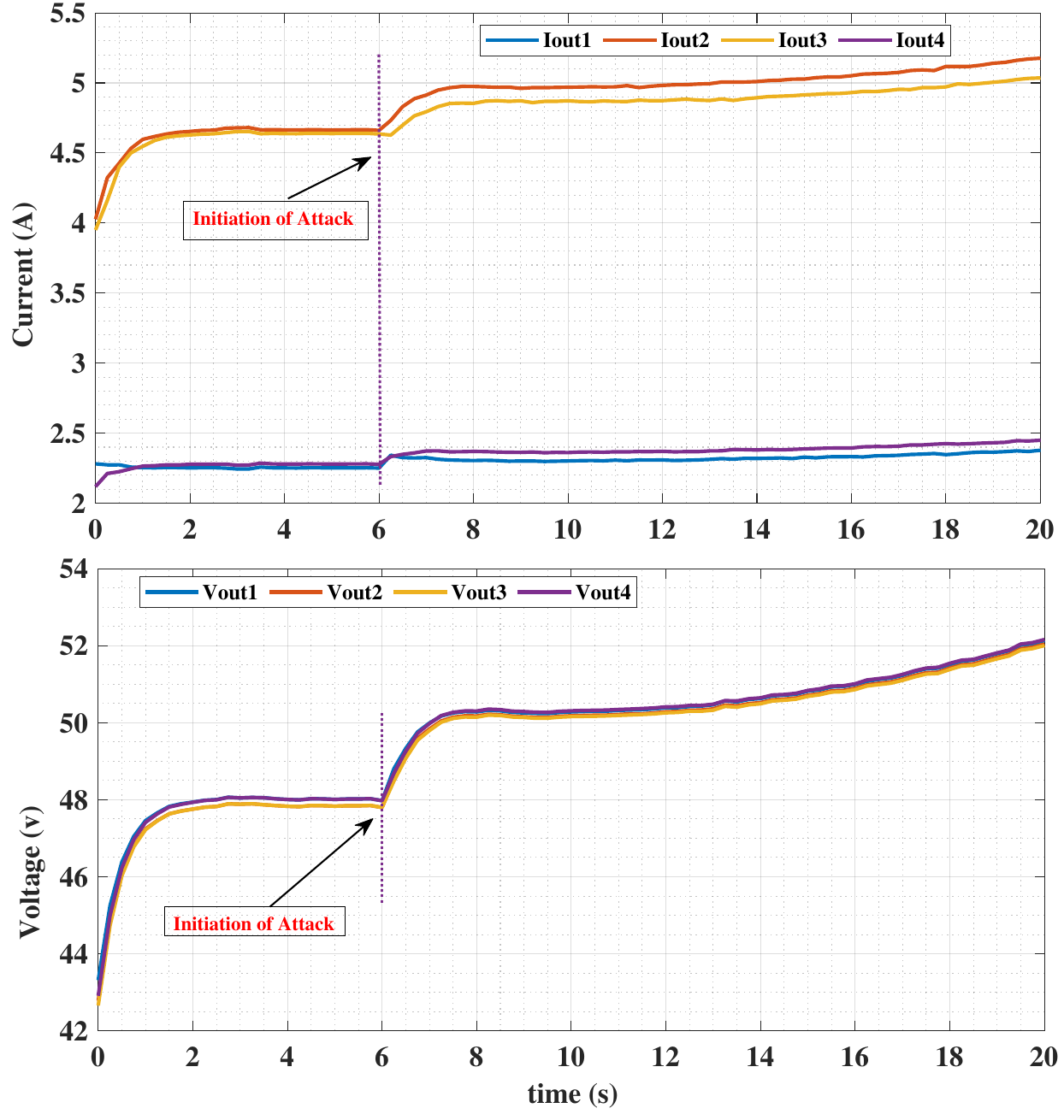}}
\caption {Performance of the proposed resilient control approach in \cite{liu2023resilient} in the case of unbounded attack: Supplied currents (top); Terminal voltages of Converters (bottom).}
\label{FIG70}
\end{figure}

\subsection{Robustness to Load Changes}
The second case study demonstrates the robustness of the proposed control strategy in response to variations in load conditions. The attacks are characterized by: to $\delta_i=\left[3~\mathrm{exp}\left(0.1t\right),4~\mathrm{exp}\left(0.2t\right),2~\mathrm{exp}\left(0.2t\right),~\mathrm{exp}\left(0.1t\right)\right]^\mathrm{T}$, $\forall i=1,2,3,4 $, and have been initiated at $t = 5.2 s$. The experiment involves altering the load resistance for converter 1 and 4 by introducing additional resistance equal to $90\Omega$ in parallel, at $t = 9 s$, followed by a halving of the resistance at $t = 14 s$.  Fig.~\ref{FIG7} illustrates the trajectories of voltage and current. Notably, as depicted in this figure, the objectives of proportional load
sharing and voltage regulation are successfully achieved, underscoring the resilience of the control strategy in accommodating uncertainties in load parameters.

\begin{figure}[ht]
\centering
{\includegraphics[width=0.4\textwidth]{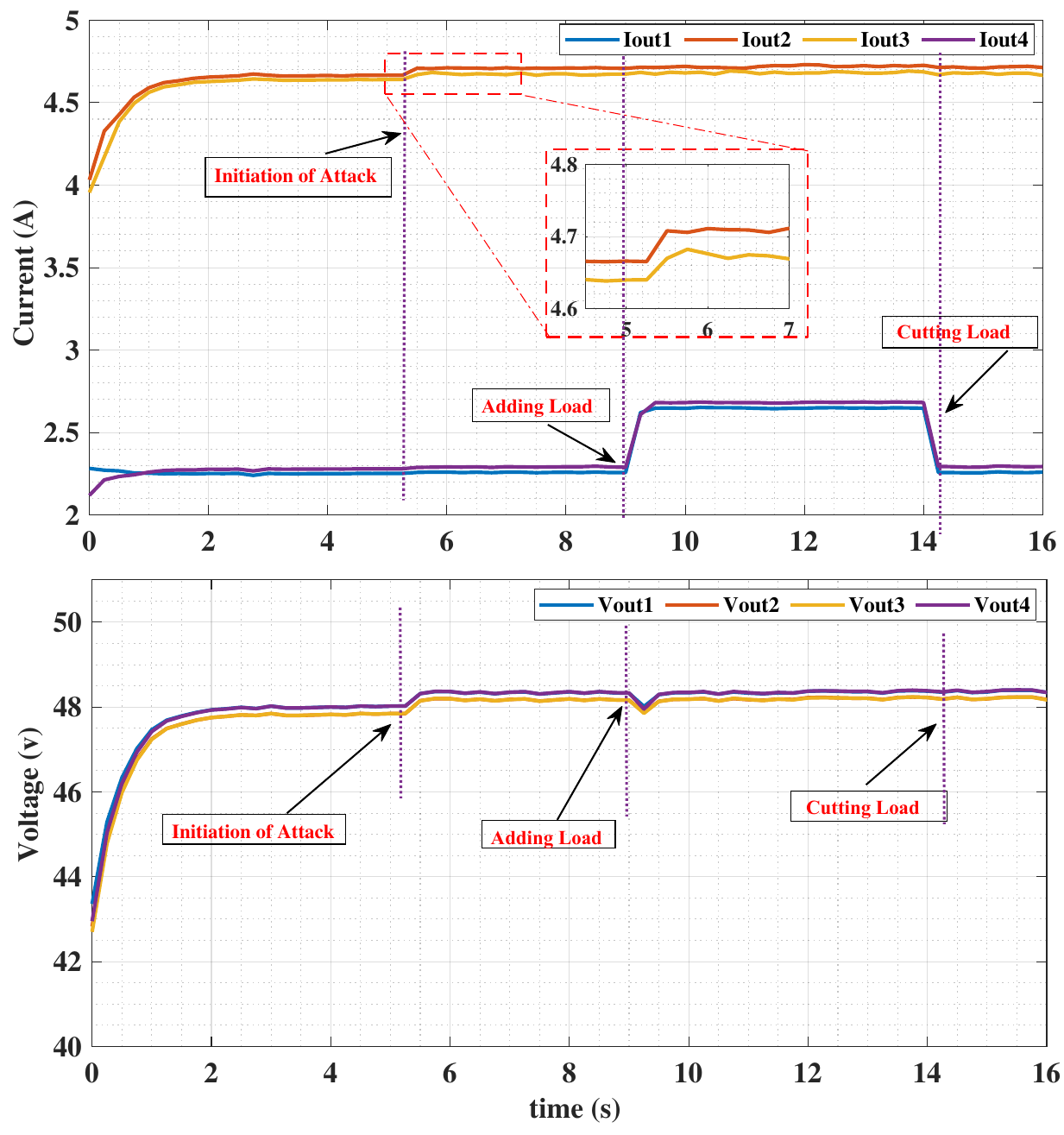}}
\caption {Performance of the proposed attack-resilient control approach in the case of unbounded attack and step load change: Supplied currents (top); Terminal voltages of Converters (bottom).}
\label{FIG7}
\end{figure}

\subsection{Communication Failure}
In this case, the influences of the communication link failure on the presented control algorithm are mainly discussed. At $t = 5.8 s$, the identical unbounded attacks mentioned in the preceding case study have been initiated, followed by the failure of communication link $2-3$ at $t = 11.2 s$. As depicted in Fig.~\ref{Fig23}, despite the failure of these links, the stability of frequency and voltage are maintained as the communication graph remains connected. The robustness of the proposed distributed secondary control scheme is evident in its ability to withstand communication link failure, ensuring system stability as long as the failure does not disrupt the graph's connectivity.
The experimental results in Section V, particularly in Fig.~\ref{FIG6} and Fig.~\ref{FIG7}, illustrate the dynamic process of secondary voltage and current adjustments. When subjected to varying loads and cyberattacks, the proposed controller recalibrates each converter's output in real-time, achieving stable operation and proportional load sharing. For instance, during the load change test depicted in Fig.~\ref{FIG7}, the system demonstrates resilience as it adjusts the output voltages and currents dynamically, ensuring continuous operation within the desired bounds despite the sudden introduction and removal of loads.

\begin{figure}[ht]
\centering
{\includegraphics[width=0.4\textwidth]{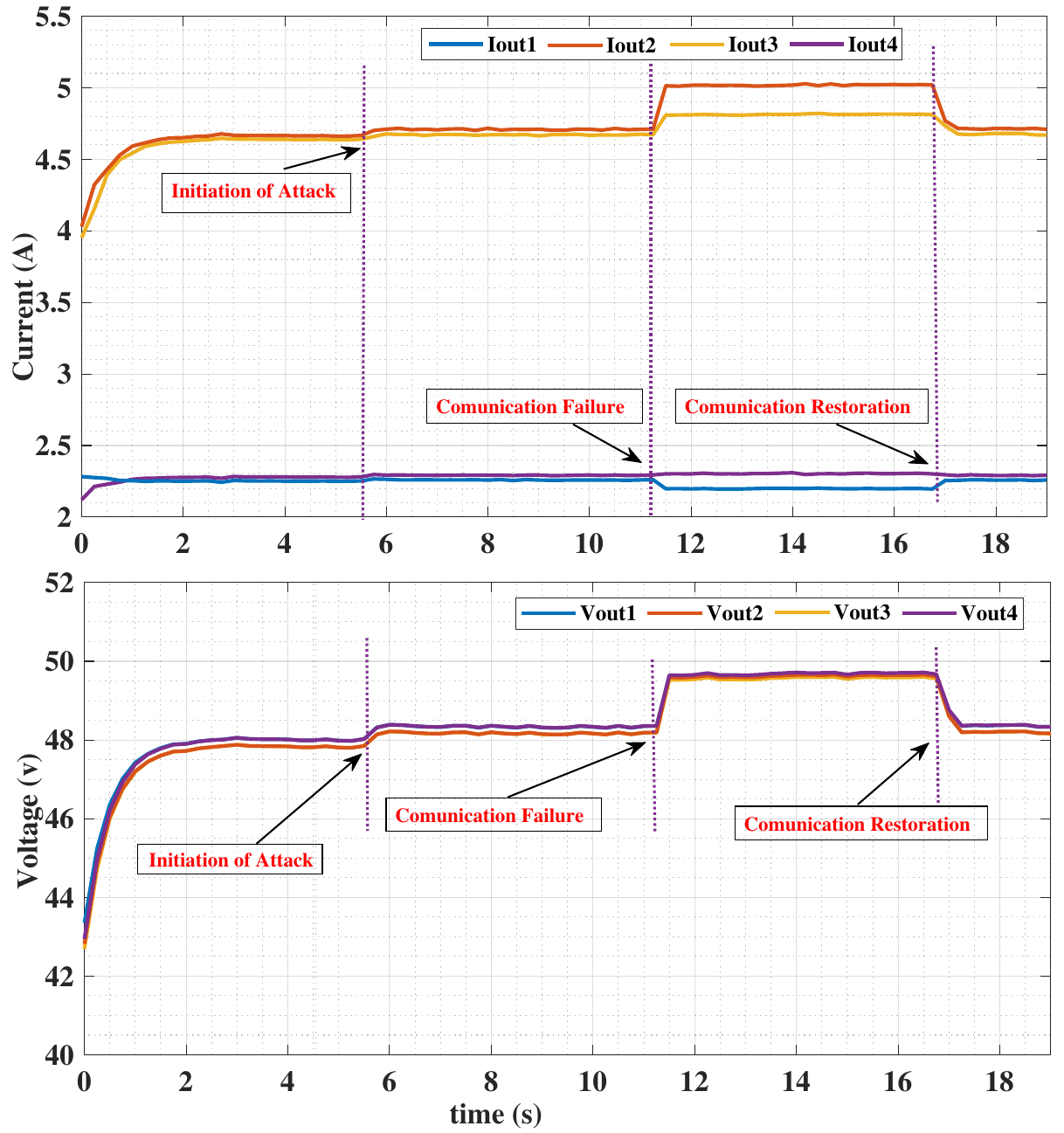}}
\caption { Performance of the proposed attack-resilient control approach in the case of unbounded attack and communication failure: Supplied currents (top); Terminal voltages of Converters (bottom).}
\label{Fig23}
\end{figure}

\section{Conclusion}
\label{sec:Conclusion}
This study has proposed a novel, fully distributed exponentially attack-reilient control framework for DC MG, tailored to enhance resilience against EU-FDI attacks, a significant challenge in the quantum era. This framework empowers DC MG to withstand EU-FDI attacks, which traditional attack-resilient control systems are not equipped to counter. The core of our approach is a consensus-based resilient secondary control implemented for each converter, specifically designed to combat EU-FDI attacks. A rigorous proof of Lyapunov stability analysis has verified that our strategy ensures UUB convergence under the EU-FDI attacks condition. The efficacy and enhanced resilience of our proposed control protocal have been further corroborated through comprehensive HIL experiments, demonstrating its practical applicability and resilience in mitigating the intricate challenges posed by quantum-era cyberattacks.
\ifCLASSOPTIONcaptionsoff
  \newpage
\fi
\bibliographystyle{IEEEtran}

\bibliography{main}

\begin{thebibliography}{10}
\providecommand{\url}[1]{#1}
\csname url@samestyle\endcsname
\providecommand{\newblock}{\relax}
\providecommand{\bibinfo}[2]{#2}
\providecommand{\BIBentrySTDinterwordspacing}{\spaceskip=0pt\relax}
\providecommand{\BIBentryALTinterwordstretchfactor}{4}
\providecommand{\BIBentryALTinterwordspacing}{\spaceskip=\fontdimen2\font plus
\BIBentryALTinterwordstretchfactor\fontdimen3\font minus \fontdimen4\font\relax}
\providecommand{\BIBforeignlanguage}[2]{{%
\expandafter\ifx\csname l@#1\endcsname\relax
\typeout{** WARNING: IEEEtran.bst: No hyphenation pattern has been}%
\typeout{** loaded for the language `#1'. Using the pattern for}%
\typeout{** the default language instead.}%
\else
\language=\csname l@#1\endcsname
\fi
#2}}
\providecommand{\BIBdecl}{\relax}
\BIBdecl

\bibitem{1}
A.~{Kwasinski}, ``Quantitative evaluation of {DC} microgrids availability: Effects of system architecture and converter topology design choices,'' \emph{IEEE Transactions on Power Electronics}, vol.~26, no.~3, pp. 835--851, March 2011.

\bibitem{am5}
Z.~Wang, W.~Wu, and B.~Zhang, ``A distributed control method with minimum generation cost for {DC} microgrids,'' \emph{IEEE Transactions on Energy Conversion}, vol.~31, no.~4, pp. 1462--1470, 2016.

\bibitem{am20}
\BIBentryALTinterwordspacing
S.~Peyghami, H.~Mokhtari, and F.~Blaabjerg, ``Chapter 3 - hierarchical power sharing control in {DC} microgrids,'' in \emph{Microgrid}, M.~S. Mahmoud, Ed.\hskip 1em plus 0.5em minus 0.4em\relax Butterworth-Heinemann, 2017, pp. 63--100. [Online]. Available: \url{https://www.sciencedirect.com/science/article/pii/B9780081017531000036}
\BIBentrySTDinterwordspacing

\bibitem{am38}
S.~Anand, B.~G. Fernandes, and J.~Guerrero, ``Distributed control to ensure proportional load sharing and improve voltage regulation in low-voltage {DC} microgrids,'' \emph{IEEE Transactions on Power Electronics}, vol.~28, no.~4, pp. 1900--1913, 2013.

\bibitem{9}
S.~{Sahoo} and S.~{Mishra}, ``A distributed finite-time secondary average voltage regulation and current sharing controller for {DC} microgrids,'' \emph{IEEE Transactions on Smart Grid}, vol.~10, no.~1, pp. 282--292, January 2019.

\bibitem{am21}
Q.~Xu, Y.~Xu, Z.~Xu, L.~Xie, and F.~Blaabjerg, ``A hierarchically coordinated operation and control scheme for {DC} microgrid clusters under uncertainty,'' \emph{IEEE Transactions on Sustainable Energy}, vol.~12, no.~1, pp. 273--283, 2020.

\bibitem{am15}
J.~M. Guerrero, J.~C. Vasquez, J.~Matas, L.~G. de~Vicuna, and M.~Castilla, ``Hierarchical control of droop-controlled {AC} and {DC} microgrids—a general approach toward standardization,'' \emph{IEEE Transactions on Industrial Electronics}, vol.~58, no.~1, pp. 158--172, 2011.

\bibitem{7}
V.~{Nasirian}, S.~{Moayedi}, A.~{Davoudi}, and F.~L. {Lewis}, ``Distributed cooperative control of {DC} microgrids,'' \emph{IEEE Transactions on Power Electronics}, vol.~30, no.~4, pp. 2288--2303, April 2015.

\bibitem{yang2022pi}
Q.~Yang, Y.~Chen, Y.~Lin, X.~Chen, and J.~Wen, ``{PI} consensus-based integrated distributed control of {MMC-MTDC} systems,'' \emph{IEEE Transactions on Power Systems}, 2022.

\bibitem{sadabadi2022robust}
M.~S. Sadabadi, N.~Mijatovic, and T.~Dragi{\v{c}}evi{\'c}, ``A robust cooperative distributed secondary control strategy for {DC} microgrids with fewer communication requirements,'' \emph{IEEE Transactions on Power Electronics}, vol.~38, no.~1, pp. 271--282, 2022.

\bibitem{liu2011false}
Y.~Liu, P.~Ning, and M.~K. Reiter, ``False data injection attacks against state estimation in electric power grids,'' \emph{ACM Transactions on Information and System Security (TISSEC)}, vol.~14, no.~1, pp. 1--33, 2011.

\bibitem{kosut2011malicious}
O.~Kosut, L.~Jia, R.~J. Thomas, and L.~Tong, ``Malicious data attacks on the smart grid,'' \emph{IEEE Transactions on Smart Grid}, vol.~2, no.~4, pp. 645--658, 2011.

\bibitem{liang2016review}
G.~Liang, J.~Zhao, F.~Luo, S.~R. Weller, and Z.~Y. Dong, ``A review of false data injection attacks against modern power systems,'' \emph{IEEE Transactions on Smart Grid}, vol.~8, no.~4, pp. 1630--1638, 2016.

\bibitem{26}
O.~A. {Beg}, T.~T. {Johnson}, and A.~{Davoudi}, ``Detection of false-data injection attacks in cyber-physical {DC} microgrids,'' \emph{IEEE Transactions on Industrial Informatics}, vol.~13, no.~5, pp. 2693--2703, October 2017.

\bibitem{28}
S.~{Sahoo}, S.~{Mishra}, J.~C. {Peng}, and T.~{Dragičević}, ``A stealth cyber-attack detection strategy for {DC} microgrids,'' \emph{IEEE Transactions on Power Electronics}, vol.~34, no.~8, pp. 8162--8174, August 2019.

\bibitem{37}
S.~{Abhinav}, H.~{Modares}, F.~L. {Lewis}, and A.~{Davoudi}, ``Resilient cooperative control of {DC} microgrids,'' \emph{IEEE Transactions on Smart Grid}, vol.~10, no.~1, pp. 1083--1085, January 2019.

\bibitem{lu2022generalized}
J.~Lu, X.~Zhang, X.~Hou, and P.~Wang, ``Generalized extended state observer-based distributed attack-resilient control for {DC} microgrids,'' \emph{IEEE Transactions on Sustainable Energy}, 2022.

\bibitem{am29}
S.~Zuo, T.~Altun, F.~L. Lewis, and A.~Davoudi, ``Distributed resilient secondary control of {DC} microgrids against unbounded attacks,'' \emph{IEEE Transactions on Smart Grid}, vol.~11, no.~5, pp. 3850--3859, 2020.

\bibitem{zuo2020resilient}
S.~Zuo, O.~A. Beg, F.~L. Lewis, and A.~Davoudi, ``Resilient networked ac microgrids under unbounded cyber attacks,'' \emph{IEEE Transactions on Smart Grid}, vol.~11, no.~5, pp. 3785--3794, 2020.

\bibitem{leng2022projections}
M.~Leng, S.~Sahoo, F.~Blaabjerg, and M.~Molinas, ``Projections of cyberattacks on stability of dc microgrids—modeling principles and solution,'' \emph{IEEE Transactions on Power Electronics}, vol.~37, no.~10, pp. 11\,774--11\,786, 2022.

\bibitem{zhou2023distributed}
D.~Zhou, Q.~Zhang, F.~Guo, Z.~Lian, J.~Qi, and W.~Zhou, ``Distributed resilient secondary control for islanded dc microgrids considering unbounded fdi attacks,'' \emph{IEEE Transactions on Smart Grid}, 2023.

\bibitem{zhou2022resilient}
J.~Zhou, Q.~Yang, X.~Chen, Y.~Chen, and J.~Wen, ``Resilient distributed control against destabilization attacks in dc microgrids,'' \emph{IEEE Transactions on Power Systems}, vol.~38, no.~1, pp. 371--384, 2022.

\bibitem{jiang2021high}
Y.~Jiang, Y.~Yang, S.-C. Tan, and S.~Y.~R. Hui, ``A high-order differentiator based distributed secondary control for dc microgrids against false data injection attacks,'' \emph{IEEE Transactions on Smart Grid}, vol.~13, no.~5, pp. 4035--4045, 2021.

\bibitem{liu2023resilient}
X.-K. Liu, S.-Q. Wang, M.~Chi, Z.-W. Liu, and Y.-W. Wang, ``Resilient secondary control and stability analysis for {DC} microgrids under mixed cyber attacks,'' \emph{IEEE Transactions on Industrial Electronics}, 2023.

\bibitem{wang2022discrete}
F.~Wang, Q.~Shan, J.~Zhu, and G.~Xiao, ``Discrete-time resilient-distributed secondary control strategy against unbounded attacks in polymorphic microgrid,'' \emph{Frontiers in Energy Research}, vol.~10, p. 961488, 2022.

\bibitem{huang2016shepherd}
Z.~Huang, T.~Zhu, Y.~Gu, and Y.~Li, ``Shepherd: sharing energy for privacy preserving in hybrid ac-dc microgrids,'' in \emph{Proceedings of the Seventh International Conference on Future Energy Systems}, 2016, pp. 1--10.

\bibitem{yuan2023distributed}
Q.-F. Yuan, Y.-W. Wang, X.-K. Liu, and Z.-W. Liu, ``Distributed privacy-preserving secondary control for dc microgrids via state decomposition,'' \emph{IEEE Transactions on Sustainable Energy}, 2023.

\bibitem{hussain2016resilient}
A.~Hussain, V.-H. Bui, and H.-M. Kim, ``A resilient and privacy-preserving energy management strategy for networked microgrids,'' \emph{IEEE Transactions on Smart Grid}, vol.~9, no.~3, pp. 2127--2139, 2016.

\bibitem{zhou2020privacy}
Q.~Zhou, M.~Shahidehpour, A.~Alabdulwahab, and A.~Abusorrah, ``Privacy-preserving distributed control strategy for optimal economic operation in islanded reconfigurable microgrids,'' \emph{IEEE Transactions on Power Systems}, vol.~35, no.~5, pp. 3847--3856, 2020.

\bibitem{albaker2018privacy}
A.~Albaker, A.~Majzoobi, G.~Zhao, J.~Zhang, and A.~Khodaei, ``Privacy-preserving optimal scheduling of integrated microgrids,'' \emph{Electric Power Systems Research}, vol. 163, pp. 164--173, 2018.

\bibitem{zuo2020distributed}
S.~Zuo, T.~Altun, F.~L. Lewis, and A.~Davoudi, ``Distributed resilient secondary control of dc microgrids against unbounded attacks,'' \emph{IEEE Transactions on Smart Grid}, vol.~11, no.~5, pp. 3850--3859, 2020.

\bibitem{nasirian2014distributed}
V.~Nasirian, S.~Moayedi, A.~Davoudi, and F.~L. Lewis, ``Distributed cooperative control of dc microgrids,'' \emph{IEEE Transactions on Power Electronics}, vol.~30, no.~4, pp. 2288--2303, 2014.

\bibitem{fax2004information}
J.~A. Fax and R.~M. Murray, ``Information flow and cooperative control of vehicle formations,'' \emph{IEEE transactions on automatic control}, vol.~49, no.~9, pp. 1465--1476, 2004.

\bibitem{zuo2022adaptive}
S.~Zuo, Y.~Zhang, and Y.~Wang, ``Adaptive resilient control of ac microgrids under unbounded actuator attacks,'' \emph{Energies}, vol.~15, no.~20, p. 7458, 2022.

\bibitem{khalil2002nonlinear}
H.~Khalil, ``Nonlinear systems,'' 2002.

\bibitem{altafini2019dynamical}
C.~Altafini, ``A dynamical approach to privacy preserving average consensus,'' in \emph{2019 IEEE 58th Conference on decision and control (CDC)}.\hskip 1em plus 0.5em minus 0.4em\relax IEEE, 2019, pp. 4501--4506.

\bibitem{krstic1995nonlinear}
M.~Krstic, P.~V. Kokotovic, and I.~Kanellakopoulos, \emph{Nonlinear and adaptive control design}.\hskip 1em plus 0.5em minus 0.4em\relax John Wiley \& Sons, Inc., 1995.

\end{thebibliography}

\end{document}